\documentclass[english]{amsart}
\usepackage[T1]{fontenc}
\usepackage[latin9]{inputenc}
\usepackage{geometry}
\usepackage{array}
\usepackage{longtable}
\usepackage{graphicx}
\usepackage{amssymb}
\usepackage{verbatim} 
\usepackage{epstopdf}

\makeatletter

\providecommand{\tabularnewline}{\\}

 \theoremstyle{definition}
 \newtheorem{defn}{Definition}
  \theoremstyle{plain}
  
  \theoremstyle{plain}
  \newtheorem*{lem*}{Lemma}
  \theoremstyle{plain}
  \newtheorem*{thm*}{Theorem}
  \theoremstyle{plain}
  \newtheorem*{algorithm*}{Algorithm}
  \newtheorem{algorithm}{Algorithm}
  \theoremstyle{definition}
  \newtheorem*{example*}{Example}
  \theoremstyle{remark}
  \newtheorem*{conclusion*}{Conclusion}
 \theoremstyle{remark}
  \newtheorem*{acknowledgement*}{Acknowledgement}

\makeatother

\usepackage{babel}

\begin{document}

\title{A new technique for text data compression}

\author{Udita Katugampola \\ southern illinois university, Carbondale, IL 62901 }

\begin{abstract}
In this paper we use ternary representation of numbers for compressing text 
data. We use a binary map for ternary digits and introduce a way to use the
 binary 11-pair, which has never been use for coding data before, and we 
 futher use  4-Digits ternary representation of alphabet with lowercase and 
 uppercase with some extra symbols that are most commonly used in day to day 
 life. We find a way to minimize the length of the bits string, which is only 
 possible in ternary representation thus drastically reducing the length of 
 the code. We also find some connection between this technique of coding data 
 and Fibonacci numbers.
\end{abstract}
\maketitle


\section*{Introduction}

Ternary or trinary is the base-3 numeral system. 
A ternary digit, \emph{trit} contains about 1.58596 ($log_{2}3$) bit of information.
Even though ternary most often refers to a system in which the three
digits, 0, 1, and 2, are all nonnegative integers, the adjective also
lends its name to the balanced ternary system, which uses -1, 0 and
+1 instead, used in comparison logic and ternary converters \cite{key-3}\cite{key-11}\cite{key-13}\cite{key-14}. 

Techniques have been developed for text document compression that are 
semiadaptive, which uses frequncy-ordered array of word-number mappings \cite{key-7}\cite{key-18}\cite{key-19}. Encoding binary digital data in ternary form has applicability to digital data 
communication systems and magnetic data storage systems \cite{key-8} and also 
in high-speed binary multipliers, which uses ternary representations of two 
numbers \cite{key-9}. Methods have been developed for converting binary signals 
into shorter balanced ternary code signals \cite{key-11}.

Variable length coding is a widely-used method in data compression, especially, 
in the applications of video data communication and storage, for example, JPEG, 
MPEG, CCITT H.261 and so on. Most of those methods implement the coding with two-field 
codes. Hsieh \cite{key-10} introduced a method using three-field representation for each code. 
Also, we can adopt ternary systems in difference coding in audio compression\cite{key-22}\cite{key-24}. 
The node of a Peano Curve can be represented by a base-3 reflected Gray Codes, RGC\cite{key-23}.

Ternary systems have been used in digital data recording systems \cite{key-20} and precoded ternary data transmission systems as well \cite{key-21} from which they can send and store more data compared to what binary system does. 

Digital data compression is an important tool because it can be utilized, for example, to reduce the storage requirements for files, to increase the rate at which data can be transferred over bandwidth limited communication channels, and to reduce the internal redundancy of data prior to its encryption in order to provide increased security.

In this paper we use standard ternary representation for coding data. We introduce a new representation called $Base$$\: B_{23}$, which uses both ternary and binary representation with the maximum use of both binary and 
ternary features in a single coded data. Before going further, let us compare the binary and ternary representaions of few numbers,

\begin{center}
$85_{10}=1010101_{2}=$$10010_{3}$ \mbox{ and} $150_{10}=10010110_{2}=12120_{3}$
\end{center}

\noindent 
Thus, the binary represenatation of 85 uses 7 bits while ternary representation uses only 
5 trits, and in 150, it uses 5 trits again, which is more than 40\% saving in memory. Here 
we also notice that the 12-pair occurs twice in this representation. Thus if we can use a 
better representation which can code the 12-pair we can save more memory when saving these 
data and can save more time when sending to another detector. Thus our main goal is to develpoe 
a system which uses ternary representation in a sophisticated mannar. That is what we are going to 
do in the rest of this paper.

\section*{The Basic Definitions}

We begin with the following two definitions.

\begin{defn}
Let $n=\epsilon_{1}\epsilon_{2}\epsilon_{3}...\epsilon_{k},$ where
$\epsilon_{i}=$$0,$1, 2; $k\in N$ be the ternary representation
of the decimal number $n.$ Then we say $n$ is in the form $A_{23}$,
if there is a map $\varphi:\{0,1,2\}\longmapsto\{00,01,10\}$ such
that,
\end{defn}
\begin{center}
\[ \phi(\epsilon_i)= \begin{cases} 
                              00, &\text{if $\epsilon_i=0$;}\\ 
                              01, &\text{if $\epsilon_i=1$;}\\
                              10, &\text{otherwise.} 
                     \end{cases} 
\]
\par\end{center}

\begin{flushleft}
for each $i=1,2,3,...,k.$
\par\end{flushleft}

$ $

\begin{flushleft}
With this definition, we can write $85_{10}=10010_{3}=0100000100_{A_{23}}$
and $150_{10}=12120_{3}=0110011000_{A_{23}}$. Thus it double the
lenth of the string which represents the number in $A_{23}$ base.
But in this format we waste the 11 pair. So we need to modify the
coding so that we can make use of the binary string $11_{2}$. Thus
we come up with the following definition, with the so called $Base\:$$B_{23}.$
\par\end{flushleft}

$ $

\begin{defn}
Let $n=\epsilon_{1}\epsilon_{2}\epsilon_{3}...\epsilon_{k},$ where
$\epsilon_{i}=$$0,$1, 2; $k\in N$ be the ternary representation
of the decimal number $n.$ Then we say $n=\varepsilon_{1}\varepsilon_{2}\varepsilon_{3}...\varepsilon_{l},$
where $\varepsilon_{i}=00,01,10,11;$ $l\leq k,$ is in Base $B_{23}$,
if there is a map $\psi:\{0,1,2\}\longmapsto\{00,01,10,11\}$ such
that $\psi\left(\epsilon_{i}\right)=\varepsilon_{j,}$where
\end{defn}
\begin{center}
\[ \varepsilon_j= \begin{cases} 
                              00, &\text{if $\epsilon_i=0$;}\\ 
                              01, &\text{if $\epsilon_i=1\,and\, \epsilon_{i+1}\neq 2$;}\\
                              10, &\text{if $\epsilon_i=2\,and\, \epsilon_{i-1}\neq 1$;}\\
                                11, &\text{if $\epsilon_i\epsilon_{i+1}=12$.} 
                     \end{cases} 
\]
\par\end{center}

\begin{flushleft}
for each $i=1,2,3,...,k.$
\par\end{flushleft}

$ $

\begin{flushleft}
That is, we say a number $n$ is in base $B_{23}$, if each of the
trits of its ternary representation is replaced by the following binary
bits in that order: 
\par\end{flushleft}

\begin{center}
0 $\mapsto$00,
\par\end{center}

\begin{center}
12 $\mapsto$11,
\par\end{center}

\begin{center}
1 $\mapsto$01,
\par\end{center}

\begin{center}
2 $\mapsto$10.
\par\end{center}

\begin{flushleft}
So we can write $85_{10}=0100000100_{A_{23}}=0100000100_{B_{23}}$
and $150_{10}=0110011000_{A_{23}}=111100_{B_{23}}$. Thus it drastically
reduces the length of the bits string if the $12_{3}$ is present
in the coding. The more $12_{3}$ pairs are present, the more compact
the code can be. Therefore, it is natural to seek the availability
of $12_{3}$ pairs in a string of ternary representation of a data.
Thus we have the following lemma,
\par\end{flushleft}

\begin{lem*}
Golden Lemma 
\end{lem*}
 The number $S_n$ of ways that a string of trits $0,$ $1$
and $2$ of length $n$ with no $12$ pairs in the string is

\begin{center}
$S_{n}=$$\frac{1}{\sqrt{5}}$$\left\{ \phi^{2n+2}-\phi^{-2n-2}\right\} $
for $n\in\mathbf{N},$ 
\par\end{center}

\begin{flushleft}
where $\phi$ is the $Golden\: Ratio$ $\frac{1+\sqrt{5}}{2}$.
\par\end{flushleft}

\begin{proof}
We prove this result using a recurrence relation. First consider the
following construction of string of $0,$1 and $2$ of size $n$. 

\noindent \begin{center}
$S_{n}=\mbox{} $\begin{tabular}{|c|c|c|c|c|c|c|c|c|c|}
\hline 
 &  &  &  & . . . &  &  &  &  & ${0\atop 1}$\tabularnewline
\hline
\end{tabular} $\longrightarrow2S_{n-1}$\vspace{.2cm}

\begin{tabular}{|c|c|c|c|c|c|c|c|c|}
\hline 
 &  &  &  & ... &  &  & $2$ & $2$\tabularnewline
\hline
\end{tabular} + \begin{tabular}{|c|c|c|c|c|c|c|c|c|}
\hline 
 &  &  &  & ... &  &  & $0$ & 2\tabularnewline
\hline
\end{tabular} $\longrightarrow$$S_{n-2}$\vspace{.2cm}

\begin{tabular}{|c|c|c|c|c|c|c|c|c|}
\hline 
 &  &  &  & ... &  & 2 & $2$ & $2$\tabularnewline
\hline
\end{tabular} + \begin{tabular}{|c|c|c|c|c|c|c|c|c|}
\hline 
 &  &  &  & ... &  & 0 & $2$ & 2\tabularnewline
\hline
\end{tabular} $\longrightarrow$$S_{n-3}$\vspace{.2cm}

\hspace{-1.4cm}\vdots\vspace{.2cm}  

\hspace{-.5cm}
\begin{tabular}{|c|c|c|c|c|c|c|c|}
\hline 
 &  & 2 & 2 & ... & 2 & 2 & 2\tabularnewline
\hline
\end{tabular} + \begin{tabular}{|c|c|c|c|c|c|c|c|}
\hline 
 &  & 0 & 2 & ... & 2 & 2 & 2\tabularnewline
\hline
\end{tabular} $\longrightarrow$$S_{2}$\vspace{.2cm}

\begin{tabular}{|c|c|c|c|c|c|c|c|}
\hline 
${0\atop 2}$ & 2 & 2 & 2 & ... & 2 & 2 & 2\tabularnewline
\hline
\end{tabular} + \begin{tabular}{|c|c|c|c|c|c|c|c|}
\hline 
 & 0 & 2 & 2 & ... & 2 & 2 & 2\tabularnewline
\hline
\end{tabular} $\longrightarrow$$S_{1}+2$$\:$

\par\end{center}
\[
\]
\begin{center}
\mbox{Figure 1: Deriving Formula}
\end{center}

\begin{flushleft}
$\vphantom{}$According to the diagram we have,
\begin{equation*}
   S_n = 2S_{n-1}+S_{n-2}+S_{n-3}+...+S_2+S_1+2,
\end{equation*}
\par\end{flushleft}

\noindent with $S_{1}=$ 3 and $S_{2}=8$ for $n$$\geq3.$ This reduces to $S_{n}=3S_{n-1}-S_{n-2}$ with $S_{1}=$ 3 and $S_{2}=8$ for
$n$$\geq3.$ Solving the recurrence relation with the
given data, we have,

\noindent \begin{center}
$S_{n}=$$\frac{1}{\sqrt{5}}$$\left\{ \left(\frac{1+\sqrt{5}}{2}\right)^{2n+2}-\left(\frac{1-\sqrt{5}}{2}\right)^{2n+2}\right\} $
for $n$$\geq3$
\par\end{center}

\noindent \begin{center}
$=$$\frac{1}{\sqrt{5}}$$\left\{ \phi^{2n+2}-\phi^{-2n-2}\right\} ,$
where $\phi=$$\frac{1+\sqrt{5}}{2}$.
\par\end{center}
\end{proof}

\begin{flushleft}
The first few terms of this sequence are $S_{1}=$3, $S_{2}=8,$ $S_{3}=$21, $S_{4}=55,$ $S_{5}=144$.
These are the even terms of the Fibonachchi sequence defined by, $F_{n}=F_{n-1}+F_{n-2}$ with $F_{1}=$2, $F_{2}=3.$
\par\end{flushleft}

$ $

\begin{flushleft}
Now, consider a uniformely distributed ternary string of $0,$ $1,$
and $2$ of size $n.$ Then the following holds:
\par\end{flushleft}

\begin{thm*}
The probability of appearing at least one $12-$ pair of a string
of $0,$ 1, and $2$ of length $n$ is asymptotically 1.
\end{thm*}
\begin{proof}
If $S_{n}$ is the number of ways that a string of trits $0,$ $1$
and $2$ of size $n$ can be arranged with no $12-$pairs, then the
probability of getting at least one pair of $12$ is

\begin{align}
    P_n &= \frac{3^n-S_n}{3^n},\notag \\ 
    &= 1- \frac{1}{3^{n}\sqrt{5}} \left(\phi^{2n+2}-\phi^{-2n-2}\right)\notag     
\end{align}Therefore, ${lim\atop n\rightarrow\infty}$$P_{n}=$${lim\atop n\rightarrow\infty}\left(\frac{3^{n}-\frac{1}{\sqrt{5}}\phi^{2n+2}}{3^{n}}\right)$
= ${lim\atop n\rightarrow\infty}1-\frac{\phi^{2}}{\sqrt{5}}\left(\frac{\phi^{2}}{3}\right)^{n}$$\rightarrow1,$
since $\phi^{2}=\frac{3+\sqrt{5}}{2}$ < 3.
\end{proof}
This shows that when $n$ gets larger, the mapping $12$$\rightarrow11,$ should efficiently reduce the length of the
$B_{23}$ string. 


\section*{Data compression with $B_{23}$}

Here we discuss one of the main usage of $B_{23}$ in coding data
in day to day life. When compared to other benifits of converting
data into binary form, word processing takes the leadership. Thus
in this paper we discuss a technique that could drastically increase
the storing and sending capabilities of data using the so called $B_{23}$ternary
coding. 

\begin{algorithm*}
Let A=$\{a_{1},a_{2},a_{3},...,a_{n}\}$ be an alphabet with, $p\left(a_{i}\right)$
being the probability of $a_{i}$ appearing in the language generated
by A. Without loss of generality, suppose $p\left(a_{1}\right)\geq p\left(a_{2}\right)\geq...\geq p\left(a_{n}\right).$
Let B=$\{b_{1},b_{2},b_{3},...,b_{n}\}$ be a sequence of integers
written in ternary form. Define, $\delta:B\mapsto\mathbb{Z}$ by \begin{equation*}
    \delta(b_i)=j, \,\text{where}\,\,\text{$b_i=...12_1...12_2...12_j...$}
\end{equation*}

\begin{flushleft}
Suppose without loss of generality, $\delta(b_{1})\geq\delta(b_{2})\geq...\geq\delta(b_{n}).$
Then we define a map $\Omega:A\longmapsto B$ such that $\Omega(a_{i})=b_{i}$
for i=1,2,...,n. 
\par\end{flushleft}
\end{algorithm*}
\begin{flushleft}
We notice that this scheme genarates a coding for the language generated
by $A.$ Now, we turn to a more practical example. That is the english
alphabet and the language generated by it. 
\par\end{flushleft}

\begin{example*}
Here we are going to find a $B_{23}$ code for 26-letter english alphabet, which is different from Huffman coding of 26-letter alphabet and also is different from the ternary Huffman coding\cite{key-22}\cite{key-2}\cite{key-4}\cite{key-25}. We also include some extra characters, which are frequently used in word processing. Before going further, we need to observe some facts related to the usage of english alphabet. When we use a language, we have to use some letters more frequently than others. Particularly in English the letter $'e'$ has much highest frequency compared to the others characters
in the alphabet, while the letter $'z'$ has the lowest frequency.
Some of these facts is summaries in the follwing table accompanied
with the two charts 1 and 2. %
\footnote{Cryptographical Mathematics, Robert Edward Lewand. MAA, Washington
DC, 2000%
}
\end{example*}
\includegraphics[bb=60bp 0bp 560bp 450bp,width=7cm,height=7cm]{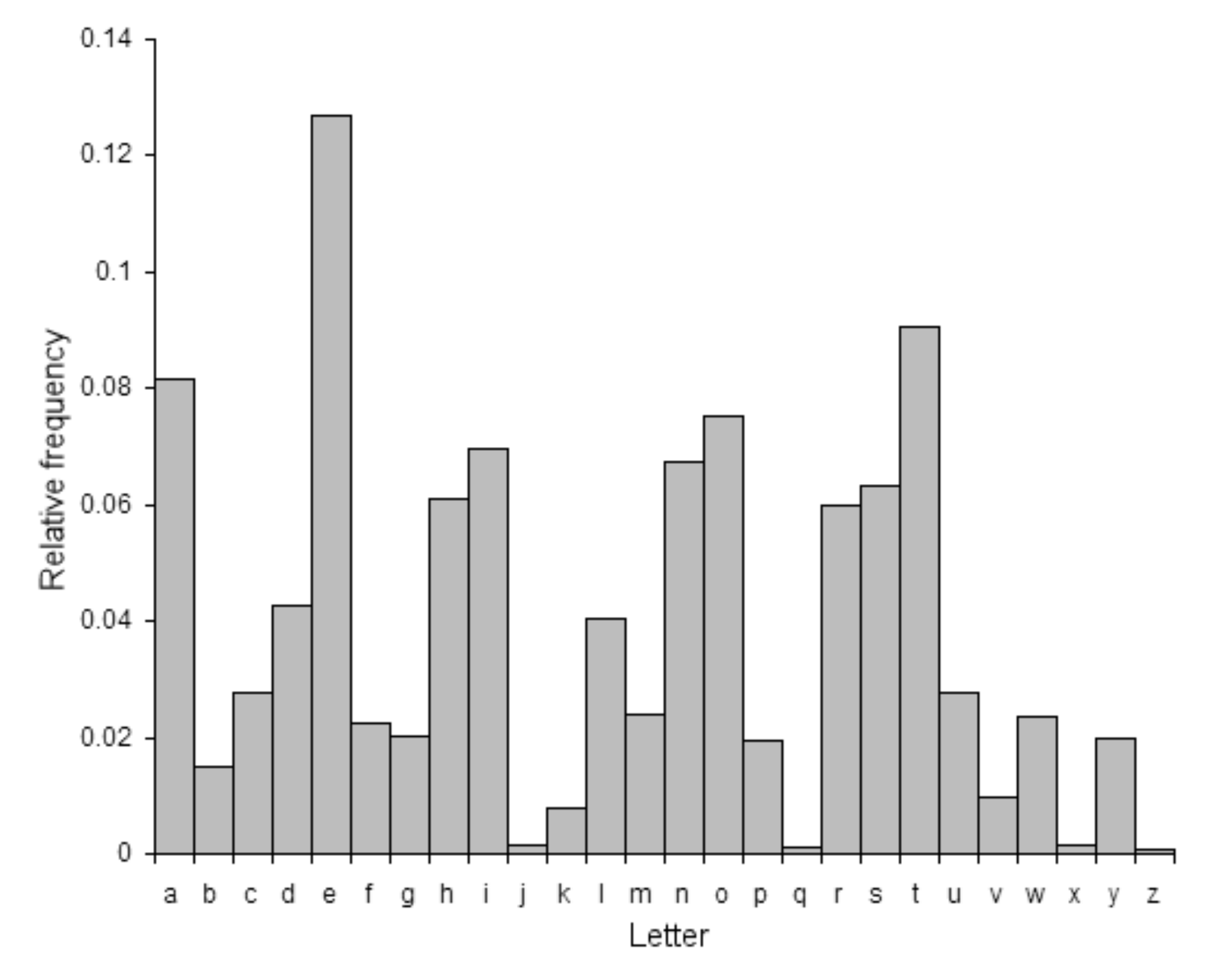}\includegraphics[bb=0bp 0bp 380bp 272bp,width=8cm,height=7cm]{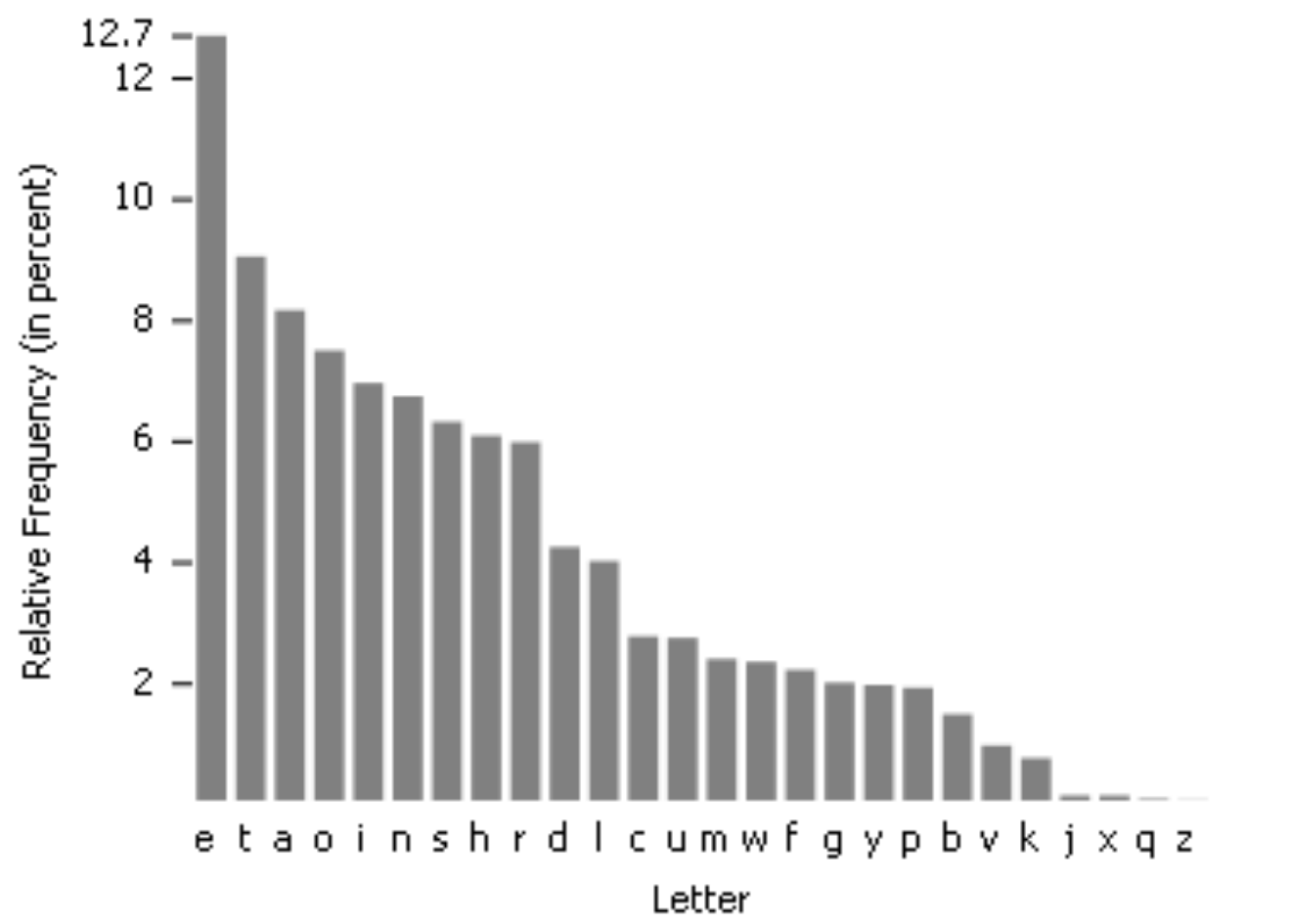}

\begin{center}
\mbox{Figure 2: Bar-chart of frequency order of 26-letter English alphabet}
\end{center}
\[
\]
According to the table below, the letters, $e,\: t,\: a,\: o,\: i,\: n,\: s,\: h,$
and $r$ has comparatively high frequency than the characters $p,\: b,\: v,\: k,\: j,\: x,\: q,$
and $z.$Therefore if we can use a shoter code for the high frequence
letters and comparatively longer code for low frequency characters
we can possible reduce the string lenth for the coding sequence. This
is the whole goal in the rest of the paper. 

\begin{longtable}{|>{\centering}p{2.5cm}|>{\centering}p{2.5cm}|>{\centering}p{2.5cm}|>{\centering}m{2.5cm}|}
\hline 
Letter & Frequency(\%) & Letter & Frequency(\%)\tabularnewline
\hline
\endhead
a & 8.167 & n & 6.749\tabularnewline
\hline 
b & 1.492 & o & 7.507\tabularnewline
\hline 
c & 2.782 & p & 1.929\tabularnewline
\hline 
d & 4.253 & q & 0.095\tabularnewline
\hline 
e & 12.702 & r & 5.987\tabularnewline
\hline 
f & 2.228 & s & 6.327\tabularnewline
\hline 
g & 2.015 & t & 9.056\tabularnewline
\hline 
h & 6.094 & u & 2.758\tabularnewline
\hline 
i & 6.966 & v & 0.978\tabularnewline
\hline
j & 0.153 & w & 2.360\tabularnewline
\hline 
k & 0.772 & x & 0.150\tabularnewline
\hline 
l & 4.025 & y & 1.974\tabularnewline
\hline 
m & 2.406 & z & 0.074\tabularnewline
\hline 
\end{longtable}
\begin{center}
\mbox{Table 1: Frequency order of 26-letter English alphabet}
\end{center}
\[
\]

Above table force us to reorder the alphabet so that we use a short
code for high frequency letter, while longer code for low frequncy
letters. Before doing this we have to notice few things that are characteristics
to the English alphabet. One major thing is alphabet has only 26 characters.
Thus if we use a ternary coding with three trits, we can code all
the 26 characters, but if we use binary coding we have to use 5 bits
to handle these characters. This is the major observation in ternary
system compared to binary system.

\section*{3 Exchanged Character Map - 3ECM}

To represent high frequency characters with short code, we simply exchanged
the positions of the English alphabet. Letters $e,\: t,\: a,\: o,\:...$
get shoter codes while the letters $v,\: k,\: x,\: q,$ and $z$ get
comparatively longer codes as in the Huffman-encoding. 
What remains is to consider which upper case letters, typically the first letter of a word, appear
more frequently as the first letter of a word. The top ten letters with frequencies, which occur at the beginning of words are:

\begin{center}
$\vphantom{}$
\par\end{center}

\begin{center}
\begin{tabular}{|c||c|c|cc|c|c|c|c|c|c|}
\hline 
Letter & T & A & I & S & O & C & M & F & P & W\tabularnewline
\hline 
Frequency(\%) & 15.94 & 15.5 & 8.23 & 7.75 & 7.12 & 5.97 & 4.26 & 4.08 & 4.0 & 3.82\tabularnewline
\hline
\end{tabular}
\par\end{center}

\begin{center}
\[
\mbox{Table 2: Frequency of the first letters}
\]
\end{center}


Clearly, the order differs from that for lower case (cf. t \& e Vs T \& E). Thus we propose the following character map with few extra sysmbols and accompanied ternary codes:

\begin{center}
$\vphantom{}$
\par\end{center}

\begin{center}
\begin{tabular}{|c|c|c|c|c|c|c|c|c|}
\hline 
Dec & Symbol & Ternary & Dec & Symbol & Ternary & Dec & Symbol & Ternary\tabularnewline
\hline
\hline 
0 & W & 0000 & 27 & z & 1000 & 54 & ! & 2000\tabularnewline
\hline 
1 & N & 0001 & 28 & p & 1001 & 55 & \$ & 2001\tabularnewline
\hline 
2 & B & 0002 & 29 & b & 1002 & 56 & \textasciicircum{} & 2002\tabularnewline
\hline 
3 & C & 0010 & 30 & w & 1010 & 57 & \% & 2010\tabularnewline
\hline 
4 & D & 0011 & 31 & x & 1011 & 58 & $\sqrt{}$ & 2011\tabularnewline
\hline 
5 & T & 0012 & 32 & e & 1012 & 59 & , & 2012\tabularnewline
\hline 
6 & F & 0020 & 33 & f & 1020 & 60 & {*} & 2020\tabularnewline
\hline 
7 & G & 0021 & 34 & g & 1021 & 61 & / & 2021\tabularnewline
\hline 
8 & H & 0022 & 35 & v & 1022 & 62 & = & 2022\tabularnewline
\hline 
9 & P & 0100 & 36 & q & 1100 & 63 & < & 2100\tabularnewline
\hline 
10 & J & 0101 & 37 & j & 1101 & 64 & > & 2101\tabularnewline
\hline 
11 & K & 0102 & 38 & k & 1102 & 65 & @ & 2102\tabularnewline
\hline 
12 & L & 0110 & 39 & y & 1110 & 66 & \& & 2110\tabularnewline
\hline 
13 & M & 0111 & 40 & m & 1111 & 67 & ' & 2111\tabularnewline
\hline 
14 & A & 0112 & 41 & n & 1112 & 68 & {}`` & 2112\tabularnewline
\hline 
15 & O & 0120 & 42 & o & 1120 & 69 & ? & 2120\tabularnewline
\hline 
16 & I & 0121 & 43 & a & 1121 & 70 & ( & 2121\tabularnewline
\hline 
17 & S & 0122 & 44 & i & 1122 & 71 & ) & 2122\tabularnewline
\hline 
18 & R & 0200 & 45 & r & 1200 & 72 & \{ & 2200\tabularnewline
\hline 
19 & Q & 0201 & 46 & s & 1201 & 73 & \} & 2201\tabularnewline
\hline 
20 & T & 0202 & 47 & t & 1202 & 74 & {[} & 2202\tabularnewline
\hline 
21 & U & 0210 & 48 & u & 1210 & 75 & ] & 2210\tabularnewline
\hline 
22 & V & 0211 & 49 & h & 1211 & 76 & \textbackslash{} & 2211\tabularnewline
\hline 
23 & . & 0212 & 50 & Space & 1212 & 77 & ; & 2212\tabularnewline
\hline 
24 & X & 0220 & 51 & d & 1220 & 78 & : & 2220\tabularnewline
\hline 
25 & Y & 0221 & 52 & l & 1221 & 79 & + & 2221\tabularnewline
\hline 
26 & Z & 0222 & 53 & c & 1222 & 80 & - & 2222\tabularnewline
\hline
\end{tabular}
\par\end{center}

\begin{center}
\[
\mbox{Table 3: Coding Table}
\]

\par\end{center}

In our table the ternary code of almost all most frequent letters contain at least 
one $12$ pair. Applying the $B_{23}$ scheme results in the following code table.

According to the table, what we achieve here is that high frequent characters
has shorter lenth compared to the others. Let us exemplify the method. For that we use a test string and code it in two ways, using $B_{23}$ and the standed $ASCII$ code and then compare the two bitstrings generated from these techniques. 
It is done in the following way. The first Algorithm converts the text directly into $B_{23}-code$,
and Algorithms 2 and 3 convert it back into human readable characters.

\begin{center}
$\vphantom{}$

\begin{tabular}{|c|c|c|c|c|c|c|c|c|} \hline  Symbol  & Ternary  & $B_{23}$  & Symbol  & Ternary  & $B_{23}$  & Symbol  & Ternary  & $B_{23}$\tabularnewline 

\hline \hline 
W  & 0000  & 00000000 & z  & 1000  & 01000000  & !  & 2000  & 10000000\tabularnewline \hline  N & 0001  & 00000001  & p  & 1001  & 01000001 & \$  & 2001  & 10000001\tabularnewline \hline  B & 0002 & 00000002  & b  & 1002  & 01000010 & \textasciicircum{}  & 2002 & 10000010\tabularnewline \hline  C  & 0010  & 00000100  & w  & 1010  & 01000100  & \%  & 2010  & 10000100\tabularnewline \hline  D  & 0011  & 00000101  & x  & 1011  & 01000101  & $\sqrt{}$  & 2011 & 10000101\tabularnewline \hline  T  & \bf{0012}  & 000011  & e  & \bf{1012} & 010011 & ,  & \bf{2012} & 100011\tabularnewline \hline  F  & 0020  & 00001000 & f  & 1020 & 01001000 & {*} & 2020  & 10001000\tabularnewline \hline  G  & 0021  & 00001001  & g  & 1021 & 01001001  & /  & 2021 & 10001001\tabularnewline \hline  H  & 0022  & 00001010  & v & 1022  & 01001010  & =  & 2022  & 10001010\tabularnewline \hline  P  & 0100  & 00010000  & q  & 1100 & 01010000  & <  & 2100 & 10010000\tabularnewline \hline  J  & 0101 & 00010001  & j  & 1101  & 01010001  & >  & 2101 & 10010001\tabularnewline \hline  K  & 0102 & 00010010 & k  & 1102 & 01010010  & @ & 2102 & 10010010\tabularnewline \hline  L  & 0110 & 00010100 & y  & 1110  & 01010100 & \&  & 2110 & 10010100\tabularnewline \hline  M  & 0111 & 00010101 & m  & 1111  & 01010101 & '  & 2111 & 10010101\tabularnewline \hline  A  & \bf{0112} & 000111 & n  & \bf{1112}  & 010111  & {}``  & \bf{2112}  & 100111\tabularnewline \hline  O  & \bf{0120} & 001100 & o  & \bf{1120}  & 011100  & ?  & \bf{2120} & 101100\tabularnewline \hline  I  & \bf{0121} & 001101 & a  & \bf{1121}  & 011101  & (  & \bf{2121}  & 101101\tabularnewline \hline  S  & \bf{0122} & 001110 & i  & \bf{1122}  & 011110  & )  & \bf{2122}  & 101110\tabularnewline \hline  R  & 0200 & 00100000 & r  & \bf{1200} & 110000  & \{  & 2200  & 10100000\tabularnewline \hline  Q  & 0201 & 00100001 & s  & \bf{1201} & 110001  & \}  & 2201  & 10100001\tabularnewline \hline  T  & 0202 & 00100010 & t  & \bf{1202} & 110010 & {[}  & 2202  & 10100010\tabularnewline \hline  U  & 0210 & 00100100 & u  & \bf{1210}  & 110100  & ]  & 2210 & 10100100\tabularnewline \hline  V  & 0211 & 00100101 & h  & \bf{1211}  & 110101  & \textbackslash{}  & 2211  & 10100101\tabularnewline \hline  .  & \bf{0212} & 001011 & Space  & \bf{1212}  & 1111  & ; & \bf{2212}  & 101011\tabularnewline \hline  X  & 0220 & 00101000 & d  & \bf{1220}  & 111000  & :  & 2220 & 10101000\tabularnewline \hline  Y  & 0221 & 00101001 & l  & \bf{1221}  & 111001  & +  & 2221  & 10101001\tabularnewline \hline  Z  & 0222  & 00101010 & c  & \bf{1222}  & 111010 & -  & 2222  & 10101010\tabularnewline 
\hline 

\end{tabular}

\par\end{center}

\begin{center}
\[
\mbox{Table 4: Complete Coding Table}
\]

\par\end{center}


\begin{algorithm}
Coding into $B_{23}$ form. Here we assume that the text string only
consists of the characters listed in the above table. 

Let $u_{ij}$:= \{(W, 0000), (N, 0001), (B, 0002), ... ,(-, 2222)\},
i = 1, 2, ..., 81; j = 1,2

Input: Text String S. Let l=length(S); NewString ={[}],

for i=0 to l do

if $u_{i1}=characterAt(S,i)$

NewString = NewString +$u_{i2}$,

end do.
\end{algorithm}
Once we received the string to the destination we use the following
two algorithm to decode it. First one to transform it back into ternary
and then the second one to decode it to human readable code. 

\begin{algorithm}
Converting into ternary.

Input received string S. Let l=length(S/2), NewString = {[}].

For i=1 to l

if SubString(S,i,i+1) = '00' then NewString = NewString + '0',

elseif SubString(S,i,i+1) = '01' then NewString = NewString + '1'

elseif SubString(S,i,i+1) = '10' then NewString = NewString + '2'

else NewString = NewString + '12'

i = i+2;

end do
\end{algorithm}
\begin{flushleft}
Now we read this string as a four character groups and assign each
such group into a single character according to the above table. So
we use the following algoritm
\par\end{flushleft}

\begin{algorithm}
Decoding ternary into human readable characters

Let $u_{ij}$:= \{(W, 0000), (N, 0001), (B, 0002), ... ,(-, 2222)\},
i = 1, 2, ..., 81; j = 1,2

Input: Ternary String S. Let l=length(S); NewString ={[}],

for i=0 to l do

if $u_{i2}=subString(S,i,i+3)$

NewString = NewString +$u_{i1}$,

else NewString = NewString+{[}],

i=i+4;

end do.
\end{algorithm}

\subsection*{Upperbound for Compression Ratio}

  To derive an upperbound for the compression ratio, we have to make few assumptions since this compression is dynamic. Let us assume for a longer text at least the \emph{space}  has the highest frequency. In order to get a numerical value we assume the chance of \emph{space} is 50\% and the other letters \emph{$a_i$} has the probabilities, \emph{$p_i$} half of the table values. Then 

\begin{align}
Compression\: Ratio &= \frac{\mbox{Length of compressed text}}{\mbox{Length of uncompressed text}}, \notag \\
                  &= \frac{\sum_{i=1}^{n}\mbox{length of letter}\: \emph{$a_i$}\: \mbox{after compression} \times p_i}                                    {\sum_{i=1}^{n}\mbox{length of letter}\: \emph{$a_i$} \:\mbox{before compression} \times p_i},\notag \\
                  &\leq \frac{4 \times 100 + 6*\left (8.167+ ... +2.758 \right)+8*(1.492+ ... +0.074)}{8 \times 200}, \notag \\
                  &= 0.64577 \notag
\end{align}  

We can achieve much stronger compressions when we applied this scheme with the Huffman encoding. We can directly use the $B_{23}$ scheme, once the technology develops to a level where we can use ternary bitstrings for data transmission.    


\begin{flushleft}
Now we use the above algorithms to the following example, and cpmpare
the results with the familier ASCII encoder. We notice that even for
a short code we see noticable reduction of the coded string. 
\par\end{flushleft}

\begin{example*}
So, consider the text string,
\end{example*}
\begin{center}
S = {}`` This is the test message.''
\par\end{center}

\begin{flushleft}
Once we apply the code into $B_{23},$we get,
\par\end{flushleft}

\begin{align*}
CodedString\, = \;&0000111101010111101100011111011110110001111111001011010101001111111100100\\
 &1001111000111001011110101010101001111000111000101110101001001010011001011
\end{align*}

$ $

\begin{flushleft}
This can be compared with the corresponding Binary string generated
by ASCII coding, which is 25\% larger than the $B_{23,}$
\par\end{flushleft}

\begin{align*}
BinaryString = \;&1010100011010000110100101110011001000000110100101110011001000000111010001\\&1010000110010100100000011101000110010101110011011101000010000001101101011\\
&00101011100110111001101100001011001110110010100101110
\end{align*}

\begin{flushleft}
$ $
\par\end{flushleft}

\begin{flushleft}
After using second algorithm, we end up getting,
\par\end{flushleft}

\begin{center}
DecodedString = \char`\"{}This is the test message.\char`\"{}
\par\end{center}

$ $

Thus if we can adopt this technique in word processing and data compression we can drastically reduce the memory needed to store information and also in data transmission.

We can extend this technique for six-digits ternary system with more characters than in this case. That would be the nest task. We can also extend this technique with slight modifications for compressing highly randomized data \cite{key-16}. We conclude the paper with the following question,\\

\paragraph*{Question:}

What kind of distributions has more $12-$pairs in a string of $0,1$,
and 2?

\begin{acknowledgement*}
Author would like to express his heart felt gratitude to Dr. Jerzy
Kocik of Department of Mathematics of Southern Illinois University,
for his invaluable suggestions and continued support. 
\end{acknowledgement*}

\end{document}